\documentclass[11pt]{article}
\usepackage[english]{babel}
\usepackage{ alphalph}
\usepackage{amsthm, amsmath, amssymb, amsbsy, amscd, amsfonts, latexsym, euscript,epsf}
\newtheorem{thm}{Theorem}[section]
\usepackage{epic,eepic}
\usepackage{texdraw}
\usepackage[dvips]{color}
\usepackage{color}
\usepackage{graphicx}
\usepackage{exscale,relsize}
\usepackage{authblk}
\usepackage{url}
\usepackage{enumerate}
\usepackage{graphicx}
\graphicspath{ {./images/} }
\allowdisplaybreaks
\newtheorem{proposition}[thm]{Proposition}
\newtheorem{corollary}[thm]{Corollary}

\newtheorem{definition}[thm]{Definition}
\newtheorem{remark}[thm]{Remark}

\newtheorem{theorem}[thm]{Theorem}

 \newcommand{\al}{\alpha } 	
  \newcommand{\be}{\beta}

\title{A discrete Darboux--Lax scheme for  integrable difference equations}
\date{}

\author[1]{X. Fisenko\thanks{xenia.fisenko9@yandex.ru}}
\author[2]{S. Konstantinou-Rizos\thanks{skonstantin84@gmail.com}}
\author[3]{P. Xenitidis\thanks{xenitip@hope.ac.uk}}
\affil[1,2]{Centre of Integrable Systems, P.G. Demidov Yaroslavl State University, Yaroslavl, Russia}
\affil[3]{School of Mathematics, Computer Science and Engineering, Liverpool Hope University, L19 9JD Liverpool, UK}

\begin{document}

\maketitle

\begin{abstract}
We propose a discrete Darboux--Lax scheme for deriving auto-B\"acklund transformations and constructing solutions to quad-graph equations that do not necessarily possess the 3D consistency property. As an illustrative example we use the Adler--Yamilov type system which is related to the nonlinear Schr\"odinger (NLS) equation \cite{SPS}. In particular, we construct an auto-B\"acklund transformation for this discrete system, its superposition principle, and we employ them in the construction of the one- and two-soliton solutions of the Adler--Yamilov system. 
\bigskip




\noindent {\phantom{\textbf{Keywords:}}} difference equations, integrable lattice equations, 3D-consistency, soliton solutions.


\end{abstract}

\section{Introduction}
 It has become understood over the past few decades that integrable systems of partial difference equations (P$\Delta$Es) are interesting in their own right, see for instance \cite{Hiet-Frank-Joshi} and references therein. On the one hand, they may model various natural phenomena, as well as processes in industry and the IT sector. On the other hand, they have many interesting algebro-geometric properties \cite{Bobenko-Suris-book, Hiet-Frank-Joshi},  and are related to many important equations of Mathematical Physics such as the Yang--Baxter equation and the tetrahedron equation \cite{ABS-2005, Caudrelier, Kassotakis-Tetrahedron, Pap-Tongas0}. Moreover, they can be derived from the discretisation of nonlinear partial differential equations (PDEs). One such approach is provided by the Darboux transformations of integrable nonlinear PDEs of evolution type \cite{ SPS, Xue-Levi-Liu}; namely, the resulting relations from the permutability of two Darboux transformations can be interpreted as a P$\Delta$E.
  
 In this paper, we focus on a special class of P$\Delta$Es, the so-called {\sl{quad-graph equations}}, or systems thereof. Quad-graph systems are equations of P$\Delta$Es defined on an elementary quadrilateral of the two-dimensional lattice. In particular, they are systems of the form
\begin{equation}\label{quad-graph1}
Q(f_{00},f_{10},f_{01},f_{11};\alpha,\beta)=0,
\end{equation}
where $Q$ is a function of its arguments $f_{ij}$, $i,j=0,1$, and may also depend on parameters $\alpha$ and $\beta$. Schematically, equation \eqref{quad-graph} can be represented on the square, where $f_{00}$, $f_{10}$, $f_{01}$ and $f_{11}$ are placed on vertices of the square and the parameters $\alpha$, $\beta$ are placed on the edges as in Figure \ref{quad-graph}. Equation \eqref{quad-graph} can be thought as a partial difference equation (P$\Delta$E) by identifying $f$ with a function of two discrete variables $n,m\in\mathbb{Z}$, and $f_{ij}$ with its shifts in the $n$ and $m$ direction, i.e. $f_{ij}=f(n+i,m+j)$.

\begin{figure}[ht]
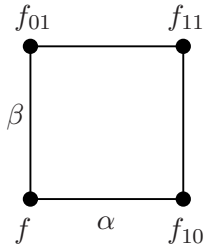

	\centering
	\centertexdraw{ 
		\setunitscale .8
		\move (0 0)  \lvec (1 0) \lvec(1 1) \lvec(0 1) \lvec(0 0)
		
		\move (0 0) \fcir f:0.0 r:0.05
		\move (1 0) \fcir f:0.0 r:0.05
		\move (0 1) \fcir f:0.0 r:0.05
		\move (1 1) \fcir f:0.0 r:0.05
		
		\htext (-.1 -.3) {$f$}
		\htext (.9 -.3) {$f_{10}$}
		\htext (-.1 1.1) {$f_{01}$}
		\htext (.9 1.1) {$f_{11}$}
		\htext (0.44 -.2) {$\al$}
		\htext (-.15 .45) {$\be$}
		
	}
	\caption{Quad--graph equation.}\label{quad-graph}
\end{figure}

Quad-graph equations have  attracted the interest of many researchers in the field of discrete integrable systems, see \cite{Hiet-Frank-Joshi} for a review. This led to the developement of methods for solving them, e.g. \cite{Atkinson2006,Wei-Fu, Hiet-Frank-Joshi, Frank-James-Hietarinta, Shi}, classifying them, e.g. \cite{ James-Maciej,Boll,PaulX}, analysing their integrability properties, e.g. \cite{Sasha-Pavlos,PaulX-2,Sarah-Frank,PaulX-3}, and relating them to the theory of Yang--Baxter and tetrahedron maps, e.g. \cite{Fordy-Pavlos-2, Pap-Tongas0}. Of particular interest are the so-called 3D-consistent equations which can be extended to the three-dimensional lattice in a consistent way. For such equations a Lax representation and a B\"acklund transformation can be constructed in a systematic manner by employing the 3D consistency of the system \cite{Hiet-Frank-Joshi}. 

In this paper, we propose a {\sl{discrete Darboux--Lax scheme}} for deriving B\"acklund transformations and constructing solutions for quad-graph equations which are not necessarily 3D consistent but have a Lax representation. More precisely, we employ gauge-type transformations of the Lax pair and demonstrate how they give rise to B\"acklund transformations for the discrete system. As an illustrative example we use the Adler--Yamilov system which was derived as a discretisation of the NLS equation via a Darboux transformation in \cite{SPS}.

The paper is organised as follows. In the next section we provide all the necessary definitions for the text to be self-contained. More specifically, after fixing our notation, we give the definitions of the integrability of a system of P$\Delta$Es in the sense of the existence of a Lax pair, and of the B\"acklund transformation. In Section \ref{sec:discrete-DL} we present the discrete Darboux--Lax scheme for constructing B\"acklund transformations and solving integrable P$\Delta$Es which do not necessarily possess the 3D consistency property. In Section \ref{sec:AY} we apply our results to the Adler--Yamilov system given in \cite{SPS} and derive the one- and two-soliton solutions by using the associated B\"acklund  transformation and the corresponding Bianchi diagram. 
Finally, the closing section contains a summary of the obtained results and a discussion on how they can be extended and generalised.

\section{Integrability of difference equations} \label{sec:preliminaries}

Let us start this section by introducing our notation. In what follows we deal with systems of equations which involve unknown functions depending on the two discrete variables $n$ and $m$. The dependence of a function $f = f(n,m)$ on these variables will be denoted with indices in the following way.
$$ f_{ij} = f(n+i,m+j), \quad {\mbox{for all }} i,j \in {\mathbb{Z}}$$
We also denote by $\mathcal{S}$ and $\mathcal{T}$ the shift operators in the $n$ and $m$ direction, respectively. Their action on a function $f$ is defined as $${\mathcal{S}}^k(f_{00})=f_{k0}, \quad {\mathcal{T}}^\ell(f_{00})=f_{0\ell}.$$
We denote vectors with bold face letters, e.g. $\textbf{f}_{00}=(f_{00}^{(1)},f_{00}^{(2)},\ldots f_{00}^{(k)})$. For systems of equations we also use bold letters. In particular, a system of quad-graph equations will be denoted by
\begin{equation} \label{quad-graph}
\textbf{Q}({\textbf{f}}_{00},\textbf{f}_{10},\textbf{f}_{01},\textbf{f}_{11})=0.
\end{equation}
Finally, we denote matrices with roman uppercase letters. For instance $\rm{L}(\mathbf{f}_{00},\mathbf{f}_{10};\al,\lambda)$ denotes a matrix with elements depending on $\mathbf{f}_{00}$, $\mathbf{f}_{10}$, $\al$ and $\lambda$. The semicolon in the arguments of the matrix is used to separate the fields $(\mathbf{f}_{00},\mathbf{f}_{10})$ from the parameters $\al$, $\lambda$. By $\lambda$ we denote the spectral parameter throughout the text.

With our notation, let ${\rm{{L}}}(\mathbf{f}_{00},\mathbf{f}_{10};\lambda)$ and ${\rm{M}}(\mathbf{f}_{00},\mathbf{f}_{01};\lambda)$ be two $k \times k$ invertible matrices which depend on a function $\bf{f}$ and the spectral parameter $\lambda$.\footnote{Matrices ${\rm{L}}$, ${\rm{M}}$ may also depend on parameters but, as they do not play any role in our discussion in this and the following section, we suppress this dependence.} Let also $\Psi=\Psi(n,m)$ be an auxiliary $k \times k$ matrix, and consider the following overdetermined linear system.
\begin{equation}\label{linear-sys}
	{\cal{S}}(\Psi)  ={\rm{L}}(\mathbf{f}_{00},\mathbf{f}_{10};\lambda) \Psi,\qquad {\cal{T}}(\Psi)={\rm{M}}(\mathbf{f}_{00},\mathbf{f}_{01};\lambda) \Psi.
\end{equation}
For given $\bf{f}$, this system has a solution $\Psi$ provided that the two equations are consistent, i.e. the compatibility condition ${\cal{T}}\left({\cal{S}}(\Psi)\right)={\cal{S}}\left({\cal{T}}(\Psi)\right)$ holds. The latter condition can be written explicitly as
\begin{equation}\label{compatibility}
	{\rm{L}}(\mathbf{f}_{01},\mathbf{f}_{11};\lambda){\rm{M}}(\mathbf{f}_{00},\mathbf{f}_{01};\lambda)={\rm{M}}(\mathbf{f}_{10},\mathbf{f}_{11};\lambda){\rm{L}}(\mathbf{f}_{00},\mathbf{f}_{10};\lambda).
\end{equation}
If the above equation holds if and only if $\bf{f}$ satisfies \eqref{quad-graph}, then we say that system of P$\Delta$Es  \eqref{quad-graph} is {\emph{integrable}}, system \eqref{linear-sys} is a {\emph{Lax pair}} for \eqref{quad-graph}, and  equation \eqref{compatibility} is called a {\emph{Lax representation}} for \eqref{quad-graph}. Moreover, matrices ${\rm{{L}}}(\mathbf{f}_{00},\mathbf{f}_{10};\lambda)$ and ${\rm{M}}(\mathbf{f}_{00},\mathbf{f}_{01};\lambda)$ are referred to as {\emph{Lax matrices}}, and without loss of generality we assume that they have constant determinants.

We close this section by giving the definition of B\"acklund transformation for quad-graph equations. Such transformations are related to the notion of integrability as it will become evident in the next section where we explore their connection to Lax pairs via the Darboux--Lax scheme.

\begin{definition}
Let ${\bf{Q}}[{\bf{f}}] := {\bf{Q}}({\bf{f}}_{00},{\bf{f}}_{10},{\bf{f}}_{01},{\bf{f}}_{11})=0$ and ${\bf{P}}[{\bf{g}}] := {\bf{P}}({\bf{g}}_{00},{\bf{g}}_{10},{\bf{g}}_{01},{\bf{g}}_{11})=0$ be two systems of quad-graph equations. Let also
\begin{equation}\label{BT-def}
\mathcal{B}(\mathbf{f}_{00},\mathbf{f}_{10},\mathbf{f}_{01},\mathbf{g}_{00},\mathbf{g}_{10},\mathbf{g}_{01} ; \varepsilon)=0
\end{equation}
be a system of  P$\Delta$Es.  If system $\mathcal{B}=0$ can be integrated for $\bf{g}$ provided that $\bf{f}$ is a solution of $\mathbf{Q}[{\bf{f}}]=0$, and the resulting ${\bf{g}}(n,m)$ is a solution to $\mathbf{P}[{\bf{g}}]=0$, and vice versa, then system \eqref{BT-def} is called a (hetero-) B\"acklund transformation for equations ${\bf{Q}}[{\bf{f}}]=0$ and ${\bf{P}}[{\bf{g}}] = 0$. If $\mathbf{Q}[{\bf{a}}] = \mathbf{P}[{\bf{a}}]$, then \eqref{BT-def} is called an auto-B\"acklund transformation for equation ${\bf{Q}}[{\bf{f}}] = 0$.
\end{definition}

\section{Discrete Darboux--Lax scheme} \label{sec:discrete-DL}
It is well known that for quad-graph systems which possess the 3D consistency property a Lax representation can be derived algorithmically \cite{Frank2002, Bobenko-Suris, Bridgman, Hiet-Frank-Joshi} and a B\"acklund transformation can be constructed systematically \cite{Atkinson2006, Hiet-Frank-Joshi}. However, there do exist quad-graph systems which have a Lax pair but do not possess the 3D consistency property. For this kind of systems we propose here a scheme for constructing Darboux and B\"acklund transformations. It should be emphasized here that this scheme works for {\emph{any}} system of difference equations irrespectively of their 3D consistency.

For the integrable quad-graph system
\begin{equation}\label{Lax-rep}
	\mathbf{Q}(\mathbf{f}_{00},\mathbf{f}_{10},\mathbf{f}_{01},\mathbf{f}_{11})=0\quad \Longleftrightarrow \quad {\rm{L}}(\mathbf{f}_{01},\mathbf{f}_{11};\lambda) {\rm{M}}(\mathbf{f}_{00},\mathbf{f}_{01};\lambda)={\rm{M}}(\mathbf{f}_{10},\mathbf{f}_{11};\lambda){\rm{L}}(\mathbf{f}_{00},\mathbf{f}_{10};\lambda),
\end{equation}
we define the discrete Darboux transformation as follows.
\begin{definition}\label{Darboux-transform}
A discrete Darboux transformation for the integrable P$\Delta$E \eqref{Lax-rep} is a gauge-like, spectral parameter-dependent transformation that leaves Lax matrices $\rm{L}$ and $\rm{M}$ covariant. That is, a transformation which  involves an invertible matrix $\rm{B}$ such that
\begin{subequations}\label{Darboux-def}
\begin{align}
   &{\rm{L}}(\mathbf{f}_{00},\mathbf{f}_{10};\lambda) \longmapsto {\rm{L}}(\tilde{\mathbf{f}}_{00},\tilde{\mathbf{f}}_{10};\lambda) = {\cal{S}}({\rm{B}})  {\rm{L}}(\mathbf{f}_{00},\mathbf{f}_{10};\lambda) {\rm{B}}^{-1},\label{Darboux-def-L}\\
    & {\rm{M}}(\mathbf{f}_{00},\mathbf{f}_{01};\lambda) \longmapsto {\rm{M}}(\tilde{\mathbf{f}}_{00},\tilde{\mathbf{f}}_{01};\lambda)= {\cal{T}}({\rm{B}})  {\rm{M}}(\mathbf{f}_{00},\mathbf{f}_{01};\lambda) {\rm{B}}^{-1}.\label{Darboux-def-M}
\end{align}
\end{subequations}
\end{definition}

A consequence of the above definition is the following proposition.

\begin{proposition}
The Darboux transformation maps fundamental solutions of the linear system
\begin{equation}\label{linear-sys-1}
{\cal{S}}(\Psi)= {\rm{L}}(\mathbf{f}_{00},\mathbf{f}_{10};\lambda) \Psi,\quad {\cal{T}}(\Psi)={\rm{M}}(\mathbf{f}_{00},\mathbf{f}_{01};\lambda) \Psi,
\end{equation}
to fundamental solutions of the linear system
\begin{equation}
{\cal{S}}(\tilde{\Psi})= {\rm{L}}(\tilde{\mathbf{f}}_{00},\tilde{\mathbf{f}}_{10};\lambda) \tilde{\Psi},\quad {\cal{T}}(\tilde{\Psi})={\rm{M}}(\tilde{\mathbf{f}}_{00},\tilde{\mathbf{f}}_{01};\lambda) \tilde{\Psi},
\end{equation}
via the relation $\tilde{\Psi}={\rm{B}} \Psi$.
\end{proposition}
\begin{proof}
Let $\Psi=\Psi(n,m)$ be a fundamental solution of the linear problem \eqref{linear-sys-1}. We set $\tilde{\Psi}={\rm{B}} \Psi$ and then we shift in the $n$ direction to find that
$$
{\cal{S}}(\tilde{\Psi})= {\cal{S}}({\rm{B}}) {\cal{S}}(\Psi) \stackrel{\eqref{linear-sys-1}}{=}{\cal{S}}({\rm{B}})  {\rm{L}}(\mathbf{f}_{00},\mathbf{f}_{10};\lambda) \Psi\stackrel{\eqref{Darboux-def-L}}{=}{\rm{L}}(\tilde{\mathbf{f}}_{00},\tilde{\mathbf{f}}_{10};\lambda) {\rm{B}} \Psi={\rm{L}}(\tilde{\mathbf{f}}_{00},\tilde{\mathbf{f}}_{10};\lambda)\tilde{\Psi}.
$$
Similarly, starting with $\tilde{\Psi}={\rm{B}} \Psi$ we shift in the $m$ direction and employ \eqref{linear-sys-1} and \eqref{Darboux-def-M} to find that ${\cal{T}}(\tilde{\Psi})={\rm{M}}(\tilde{\mathbf{f}}_{00},\tilde{\mathbf{f}}_{01};\lambda)\tilde{\Psi}$. Moreover, the solution $\tilde{\Psi}$ is fundamental, since $\Psi$ is fundamental and $\det(\tilde{\Psi})=\det{\rm{B}}\det{\Psi}\neq 0$.
\end{proof}

Using the above definition of the discrete Darboux transformation and corresponding Darboux matrix, we propose the following approach for the construction of a Darboux matrix and B{\"a}cklund transformation, as well as for the derivation of the superposition principle for the B\"acklund transformation.
\begin{itemize}
\item We start by assuming an initial form for matrix $\rm{B}$. The simplest assumption we can make is that matrix ${\rm{B}}$ depends linearly on the spectral parameter, i.e.
\begin{equation} \label{def-B}
{\rm{B}} = \lambda {\rm{B}}^{(1)} + {\rm{B}}^{(0)},
\end{equation}
where matrices $ {\rm{B}}^{(1)}$ and ${\rm{B}}^{(0)}$ do not depend on $\lambda$.

\item We determine the elements of these two matrices by employing equations \eqref{Darboux-def} written as
\begin{equation}\label{det-eqs-DB}
{\rm{L}}(\tilde{\mathbf{f}},\tilde{\mathbf{f}}_{10};\lambda) {\rm{B}}= {\cal{S}}({\rm{B}}){\rm{L}}(\mathbf{f},\mathbf{f}_{10};\lambda), \quad 
{\rm{M}}(\tilde{\mathbf{f}},\tilde{\mathbf{f}}_{01};\lambda) {\rm{B}}= {\cal{T}}({\rm{B}}){\rm{M}}(\mathbf{f},\mathbf{f}_{01};\lambda).
\end{equation}
In our calculations we also have to take into account that $\det({\rm{B}})$ is a constant, an obvious consequence of \eqref{det-eqs-DB} and our assumption that the Lax matrices have constant determinants.

\item The derived Darboux matrix will depend in general on the `old' and the `new' fields $\bf{f}$ and $\tilde{\bf{f}}$, as well as the spectral parameter $\lambda$, and a parameter $\varepsilon$. It may also depend on some auxilliary function, a potential, $g(n,m)$. That is,
$$
{\rm{B}}=\lambda {\rm{B}}^{(1)}(\mathbf{f}_{00},\tilde{\mathbf{f}}_{00},g;\varepsilon) + {\rm{B}}^{(0)}(\mathbf{f}_{00},\tilde{\mathbf{f}}_{00},g;\varepsilon).
$$
\item In the construction of the Darboux matrix \eqref{def-B} there will be some algebraic relations that define the Darboux matrix elements, as well as some difference equations for its elements. These difference equations will be of the form
\begin{equation}\label{discrete-BT}
\mathcal{B}^{(n)}(\mathbf{f}_{00},\mathbf{f}_{10},\tilde{\mathbf{f}}_{00},\tilde{\mathbf{f}}_{10},g;\varepsilon)=0, \qquad 	\mathcal{B}^{(m)}(\mathbf{f}_{00},\mathbf{f}_{01},\tilde{\mathbf{f}}_{00},\tilde{\mathbf{f}}_{01},g;\varepsilon)=0,
\end{equation}
and constitute the $n$- and the $m$-part, respectively, of an auto-B\"acklund transformation that relates the `old' and the `new' fields. In what follows, we will denote this transformation simply with ${\mathcal{B}}(\mathbf{f},\tilde{\mathbf{f}},g;\varepsilon) = 0$.

\item The Bianchi commuting diagram, aka superposition principle, for the auto-B\"acklund transformation \eqref{discrete-BT} follows from the permutation of four Darboux matrices according to the diagram in Figure \ref{fig:Bianchicomd}.

\begin{figure}[h]
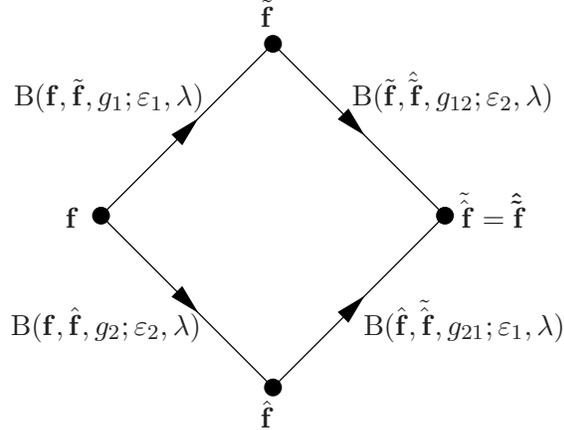

	\centertexdraw{ \setunitscale .6 \linewd 0.01 \arrowheadtype t:F
		\move(0 0) \lvec (1.5 1.5) \lvec(3 0) \lvec(1.5 -1.5) \lvec(0 0) 
		\move(0 0) \fcir f:0.0 r:0.08
		\move(1.5 1.5) \fcir f:0.0 r:0.08
		\move(3 0) \fcir f:0.0 r:0.08
		\move(1.5 -1.5) \fcir f:0.0 r:0.08
		\move(.8 .8) \avec(.85 .85)
		\move(2.25 .75) \avec(2.3 .7)
		\move(.8 -.8) \avec(.85 -.85)
		\move(2.25 -.75) \avec(2.3 -.7) 
		\htext (-0.3 -0.1) {$\bf{f}$} \htext (1.4 1.65) {$\tilde{\bf{f}}$} 
		\htext (3.15 -.1) {$\tilde{\hat{\bf{f}}} = \hat{\tilde{\bf{f}}}$} \htext (1.4 -1.85) {$\hat{\bf{f}}$}
		\htext(-0.75 0.9){${\rm{B}}(\mathbf{f},\tilde{\mathbf{f}},g_1;\varepsilon_1,\lambda)$} 
		\htext(-0.78 -1.1){${\rm{B}}(\mathbf{f},\hat{\mathbf{f}},g_2;\varepsilon_2,\lambda)$} 
		\htext(2.2 0.9){${\rm{B}}(\tilde{\mathbf{f}},\hat{\tilde{\mathbf{f}}},g_{12};\varepsilon_2,\lambda)$} 
		\htext(2.3 -1.1){${\rm{B}}(\hat{\mathbf{f}},\tilde{\hat{\mathbf{f}}},g_{21};\varepsilon_1,\lambda)$}}
	\caption{Bianchi commuting diagram. It should be noted that $g_{12} \ne g_{21}$.} \label{fig:Bianchicomd}
\end{figure}

More precisely, starting with a solution $\bf{f}$ of \eqref{Lax-rep} we can construct two new solutions $\tilde{\bf{f}}$ and $\hat{\bf{f}}$ using the B\"acklund transformations ${\mathcal{B}}(\mathbf{f},\tilde{\mathbf{f}},g_1;\varepsilon_1) =0$ and ${\mathcal{B}}(\mathbf{f},\hat{\mathbf{f}},g_2;\varepsilon_2)=0$, respectively. Then, we can use  the B\"acklund transformation with initial solution $\tilde{\bf{f}}$, a new potential $g_{12}$ and parameter $\varepsilon_2$ to derive a new solution $\hat{\tilde{\mathbf{f}}}$, i.e. ${\mathcal{B}}(\tilde{\mathbf{f}},\hat{\tilde{\mathbf{f}}},g_{12};\varepsilon_2) =0$. In the same fashion, we can start with $\hat{\bf{f}}$, a new potential $g_{21}$ and parameter $\varepsilon_1$ to derive solution $\tilde{\hat{\mathbf{f}}}$, i.e. ${\mathcal{B}}(\hat{\mathbf{f}},\tilde{\hat{\mathbf{f}}},g_{21};\varepsilon_1) =0$. By requiring $\tilde{\hat{\mathbf{f}}} = \hat{\tilde{\bf{f}}}$, according to Figure \ref{fig:Bianchicomd}, we can construct this solution algebraically and it follows from the commutativity of the corresponding Darboux matrices, 
\begin{equation}\label{BT-Bianchi}
{\rm{B}}(\tilde{\mathbf{f}},\hat{\tilde{\mathbf{f}}},g_{12};\varepsilon_2,\lambda){\rm{B}}(\mathbf{f},\tilde{\mathbf{f}},g_1;\varepsilon_1,\lambda) = {\rm{B}}(\hat{\mathbf{f}},\tilde{\hat{\mathbf{f}}},g_{21};\varepsilon_1,\lambda){\rm{B}}(\mathbf{f},\hat{\mathbf{f}},g_2;\varepsilon_2,\lambda).
\end{equation}
\end{itemize}

\section{The Adler--Yamilov system} \label{sec:AY}

From the discussion in the previous section it is already obvious that multidimensional consistency is not essential for this method of derivation of B\"acklund transformations. In this section, we demonstrate the application of our scheme using as an illustrative example the Adler--Yamilov system related to the nonlinear Schr\"odinger equation \cite{SPS}.

The Adler--Yamilov system can be written as
\begin{equation}\label{Adler-Yamilov}
 p_{10}-p_{01}-\frac{\al-\be}{1+p_{00}q_{11}}p_{00} = 0 ,\qquad q_{10}-q_{01}+\frac{\al-\be}{1+p_{00}q_{11}}q_{11} = 0,
\end{equation}
where $\al$, $\be$ are complex parameters. Moreover, using matrix
$${\rm{L}}(f,g;a,\lambda) = \lambda {\rm{L}}^{(1)} + {\rm{L}}^{(2)}(f,g;a) = \lambda \begin{pmatrix}
	1 & 0\\
	0 & 0
\end{pmatrix}
+
\begin{pmatrix}
	a +f g & f\\
	g & 1
\end{pmatrix},$$
a Lax pair for \eqref{Adler-Yamilov} can be written as
\begin{subequations}\label{Lax-pair-AY}
\begin{align}
&{\cal{S}}(\Psi)={\rm{L}}(p_{00},q_{10};\al,\lambda)\Psi=\left(\lambda \begin{pmatrix}
   1 & 0\\
   0 & 0
\end{pmatrix}
+
\begin{pmatrix}
   \al+p_{00}q_{10} & p_{00}\\
   q_{10} & 1
\end{pmatrix}
\right)\Psi,\\
&{\cal{T}}(\Psi)={\rm{L}}(p_{00},q_{01};\be,\lambda)\Psi=\left(\lambda \begin{pmatrix}
   1 & 0\\
   0 & 0
\end{pmatrix}
+
\begin{pmatrix}
   \be+p_{00}q_{01} & p_{00}\\
   q_{01} & 1
\end{pmatrix}
\right)\Psi.
\end{align}
\end{subequations}
 
\subsection{The discrete Darboux--Lax scheme for the Adler--Yamilov system} \label{sec:DL-AY}

We start with our choice \eqref{def-B} for the initial form of the Darboux matrix $\rm{B}$, 
\begin{equation}
{\rm{B}} = \lambda {\rm{B}}^{(1)} + {\rm{B}}^{(0)} = \lambda  \begin{pmatrix}
	f^{(1)}_{00} & f^{(2)}_{00} \\
	f^{(3)}_{00} & f^{(4)}_{00}
\end{pmatrix} + 
\begin{pmatrix}
	g^{(1)}_{00} & g^{(2)}_{00} \\
	g^{(3)}_{00} & g^{(4)}_{00}
\end{pmatrix},
\end{equation}
and the determining equations \eqref{det-eqs-DB}, which now become
\begin{subequations} \label{det-eqs-DB-1}
\begin{align}
&\left( \lambda {\rm{L}}^{(1)} + {\rm{L}}^{(2)}(\tilde{p}_{00},\tilde{q}_{10};\al) \right) \left( \lambda {\rm{B}}^{(1)} + {\rm{B}}^{(0)}\right)= \left( \lambda {\cal{S}}\left({\rm{B}}^{(1)}\right) + {\cal{S}}\left({\rm{B}}^{(0)}\right) \right) \left( \lambda {\rm{L}}^{(1)} + {\rm{L}}^{(2)}(p_{00},q_{10};\al)\right),\label{det-eqs-DB-1-n}\\
&\left( \lambda {\rm{L}}^{(1)} + {\rm{L}}^{(2)}(\tilde{p}_{00},\tilde{q}_{01};\be) \right) \left( \lambda {\rm{B}}^{(1)} + {\rm{B}}^{(0)}\right)= \left( \lambda {\cal{T}}\left({\rm{B}}^{(1)}\right) + {\cal{T}}\left({\rm{B}}^{(0)}\right) \right) \left( \lambda {\rm{L}}^{(1)} + {\rm{L}}^{(2)}(p_{00},q_{01};\be)\right).\label{det-eqs-DB-1-m}
\end{align}
\end{subequations}
Since all the matrices involved in \eqref{det-eqs-DB-1} are independent of $\lambda$, we collect the coefficients of the different powers of the spectral parameter. The $\lambda^2$ terms yield equations 
$${\rm{L}}^{(1)}  {\rm{B}}^{(1)} = {\cal{S}}\left( {\rm{B}}^{(1)}\right)  {\rm{L}}^{(1)}, \quad {\rm{L}}^{(1)}  {\rm{B}}^{(1)} = {\cal{T}}\left( {\rm{B}}^{(1)}\right)  {\rm{L}}^{(1)},$$
which lead to
\begin{equation}\label{eq:det-B-1}
f^{(1)}_{00} = c_1 \in {\mathbb{R}},\quad f^{(2)}_{00} = f^{(3)}_{00} = 0.
\end{equation}
The $\lambda$ terms in relations \eqref{det-eqs-DB-1} are
$${\rm{L}}^{(1)}  {\rm{B}}^{(0)} +  {\rm{L}}^{(2)}(\tilde{p}_{00},\tilde{q}_{10};\al)  {\rm{B}}^{(1)} =  {\cal{S}}\left({\rm{B}}^{(1)}\right) {\rm{L}}^{(2)}(p_{00},q_{10};\al) + {\cal{S}}\left({\rm{B}}^{(0)}\right) {\rm{L}}^{(1)},$$
$${\rm{L}}^{(1)}  {\rm{B}}^{(0)} +  {\rm{L}}^{(2)}(\tilde{p}_{00},\tilde{q}_{01};\be)  {\rm{B}}^{(1)} =  {\cal{T}}\left({\rm{B}}^{(1)}\right) {\rm{L}}^{(2)}(p_{00},q_{01};\be) + {\cal{T}}\left({\rm{B}}^{(0)}\right) {\rm{L}}^{(1)},$$
which in view of \eqref{eq:det-B-1} imply
\begin{equation}
f^{(4)}_{00} = c_2 \in {\mathbb{R}}, \quad g^{(2)}_{00} = c_1 p_{00} - c_2 \tilde{p}_{00}, \quad g^{(3)}_{00}= c_1 \tilde{q}_{00} - c_2 q_{00},\label{eq:det-B-2a}
\end{equation}
\begin{equation}
\left({\cal{S}} - 1\right)\left(g^{(1)}_{00}\right)  = c_1 \left(\tilde{p}_{00} \tilde{q}_{10} -p_{00} q_{10}\right), \quad  \left({\cal{T}} - 1\right)\left(g^{(1)}_{00}\right)  = c_1 \left( \tilde{p}_{00} \tilde{q}_{01} - p_{00} q_{01} \right).\label{eq:det-B-2b}
\end{equation}
The $\lambda$ independent terms,
\begin{equation} \label{eq:lambda-0}
{\rm{L}}^{(2)}(\tilde{p}_{00},\tilde{q}_{10};\al) {\rm{B}}^{(0)} =  {\cal{S}}\left({\rm{B}}^{(0)}\right)  {\rm{L}}^{(2)}(p_{00},q_{10};\al),\quad   {\rm{L}}^{(2)}(\tilde{p}_{00},\tilde{q}_{01};\be) {\rm{B}}^{(0)} =  {\cal{T}}\left({\rm{B}}^{(0)}\right)  {\rm{L}}^{(2)}(p_{00},q_{01};\be),
\end{equation}
determine $g^{(4)}_{00}$ and provide us with the corresponding auto-B\"acklund transformation. Specifically, the $(2,2)$-elements of the above relations yield
\begin{equation}\label{eq:det-B-3}
 \left({\cal{S}} - 1\right)\left(g^{(4)}_{00}\right)  = c_2 \left(p_{00} q_{10}-\tilde{p}_{00} \tilde{q}_{10} \right), \quad  \left({\cal{T}} - 1\right)\left(g^{(4)}_{00}\right)  = c_2 \left( p_{00} q_{01}-\tilde{p}_{00} \tilde{q}_{01} \right).
\end{equation}
The remaining entries of \eqref{eq:lambda-0} in view of \eqref{eq:det-B-2b} and \eqref{eq:det-B-3} become 
\begin{subequations} \label{eq:det-B-3a}
\begin{eqnarray}
&& c_1 \left(p_{10} + p_{00} (\al + p_{00} q_{10})\right) - c_2 \left(\tilde{p}_{10} + \tilde{p}_{00} (\al + \tilde{p}_{00} \tilde{q}_{10})\right) = g^{(4)}_{00} \tilde{p}_{00} - g^{(1)}_{00} p_{00}, \\
&& c_1 \left(p_{01} + p_{00} (\be + p_{00} q_{01})\right) - c_2 \left(\tilde{p}_{01} + \tilde{p}_{00} (\be + \tilde{p}_{00} \tilde{q}_{01})\right) = g^{(4)}_{00} \tilde{p}_{00} - g^{(1)}_{00} p_{00}, \\
&& \tilde{q}_{10} = \frac{c_1 \tilde{q}_{00} - c_2 (q_{00} - \al q_{10}) - g^{(4)}_{00} q_{10}}{c_1 (\al + p_{00} q_{10}) -c_2 \tilde{p}_{00} q_{10} - g^{(1)}_{00}},\quad  \tilde{q}_{01} = \frac{c_1 \tilde{q}_{00} - c_2 (q_{00} - \be q_{01}) - g^{(4)}_{00} q_{01}}{c_1 (\be + p_{00} q_{01}) -c_2 \tilde{p}_{00} q_{01} - g^{(1)}_{00}},
\end{eqnarray}
\end{subequations} 
which play the role of the B\"acklund transformation.

Finally we require the determinant of the Darboux matrix $\rm{B}$, which in view of \eqref{eq:det-B-1} and \eqref{eq:det-B-2a} can be written as
$\det\left({\rm{B}}\right) = c_1 c_2 \lambda^2 + \left(c_2 g^{(1)}_{00} + c_1 g^{(4)}_{00} \right) \lambda + g^{(1)}_{00} g^{(4)}_{00} -  \left( c_1 p_{00} - c_2 \tilde{p}_{00}\right) \left( c_1 \tilde{q}_{00} - c_2 q_{00}\right)$,
to be constant. This requirement implies the relations
\begin{equation}\label{eq:det-B-4}
c_2 g^{(1)}_{00} + c_1 g^{(4)}_{00} = \kappa,\quad  g^{(1)}_{00} g^{(4)}_{00} -  \left( c_1 p_{00} - c_2 \tilde{p}_{00}\right) \left( c_1 \tilde{q}_{00} - c_2 q_{00}\right) = \varepsilon, \quad \kappa, \varepsilon \in {\mathbb{R}}.
\end{equation}
It should be noted that the first relation in \eqref{eq:det-B-4} may also be viewed as a  consequence of  \eqref{eq:det-B-2b} and \eqref{eq:det-B-3}.

Summarizing, so far we have shown that the Darboux matrix $\rm{B}$  has the form
\begin{equation} \label{eq:det-B-5}
	{\rm{B}} = \lambda {\rm{B}}^{(1)} + {\rm{B}}^{(0)} = \lambda  \begin{pmatrix}
		c_1 & 0 \\
		0 & c_2
	\end{pmatrix} + 
	\begin{pmatrix}
		g^{(1)}_{00} &c_1 p_{00} - c_2 \tilde{p}_{00} \\
		c_1 \tilde{q}_{00} - c_2 q_{00} & g^{(4)}_{00}
	\end{pmatrix},
\end{equation}
where potentials $g^{(1)}$ and $g^{(4)}$ are determined by \eqref{eq:det-B-2b}, \eqref{eq:det-B-3} and \eqref{eq:det-B-4}, and the B\"acklund transformation is given by \eqref{eq:det-B-3a}.

We consider now two cases: (i) $c_1=0$ and $c_2 \ne 0$, and (ii) $c_1 \ne 0$ and $c_2 = 0$. 

\subsubsection*{First case: $c_1=0$ and $c_2 \ne 0$} \label{sec:case1}

If $c_1=0$ and $c_2 \ne 0$, then we can choose $c_2=1$ without loss of generality. Then equations \eqref{eq:det-B-2b} imply that $g^{(1)}_{00}$ is a constant and we choose $g^{(1)}_{00}= 1$.\footnote{It means we choose $\kappa=1$ in the first relation of \eqref{eq:det-B-4}.} Moreover,  the second relation in \eqref{eq:det-B-4} implies
	\begin{equation}
	g^{(4)}_{00} = \tilde{p}_{00} q_{00}+ \varepsilon, \quad \varepsilon \in {\mathbb{R}}.
	\end{equation}
With these choices, the Darboux matrix becomes
\begin{equation}\label{Darboux-matrix-AY-1}
	{\rm{B}}(q_{00},\tilde{p}_{00};\varepsilon)=\lambda
	\begin{pmatrix}
		0 & 0 \\
		0 & 1
	\end{pmatrix}
	+
	\begin{pmatrix}
		1 & -\tilde{p}_{00} \\
		-q_{00} & \varepsilon + \tilde{p}_{00} q_{00}
	\end{pmatrix},
\end{equation}
and relations \eqref{eq:det-B-3a} become
\begin{subequations}\label{BT-1}
	\begin{align}
		&  \tilde{p}_{10}=p_{00}+\frac{\al-\varepsilon}{1+\tilde{p}_{00}q_{10}}\tilde{p}_{00},
		\quad
		\tilde{q}_{10}=q_{00}-\frac{\al-\varepsilon}{1+\tilde{p}_{00}q_{10}}q_{10},\label{BT-1-n}\\
		& \tilde{p}_{01}=p_{00}+\frac{\be-\varepsilon}{1+\tilde{p}_{00}q_{01}}\tilde{p}_{00},
		\quad
		\tilde{q}_{01}=q_{00}-\frac{\be-\varepsilon}{1+\tilde{p}_{00}q_{01}}q_{01}.\label{BT-1-m}
	\end{align}
\end{subequations}
In view of the above choices and system \eqref{BT-1}, relations \eqref{eq:det-B-3} hold identically. 

System \eqref{BT-1} is an auto-B\"acklund transformation for the Adler--Yamilov system \eqref{Adler-Yamilov}. Indeed, if we shift equations \eqref{BT-1-n} in the $m$ direction, and  equations \eqref{BT-1-m} in the $n$ direction, respectively, then it can be readily verified that the two expressions for $\tilde{p}_{11}$ and the two expressions for $\tilde{q}_{11}$ coincide modulo the Adler--Yamilov system \eqref{Adler-Yamilov}. Conversely, we rearrange the above system for $p_{00}$, $q_{10}$ and $q_{01}$,
\begin{subequations}\label{BT-1-inv}
	\begin{align}
		&  q_{10}=\frac{q_{00} - \tilde{q}_{10}}{\al - \varepsilon -\tilde{p}_{00} (q_{00} - \tilde{q}_{10})},
		\quad
		p_{00}=\tilde{p}_{10} - \tilde{p}_{00} \left(\al - \varepsilon -\tilde{p}_{00} (q_{00} - \tilde{q}_{10}) \right),\label{BT-1-inv-n}\\
		&  q_{01}=\frac{q_{00} - \tilde{q}_{01}}{\be - \varepsilon -\tilde{p}_{00} (q_{00} - \tilde{q}_{01})},
\quad
p_{00}=\tilde{p}_{01} - \tilde{p}_{00} \left(\be - \varepsilon -\tilde{p}_{00} (q_{00} - \tilde{q}_{01}) \right).\label{BT-1-inv-m}
	\end{align}
\end{subequations}
If we shift the first equation in \eqref{BT-1-inv-n} in the $m$ direction and the first equation in \eqref{BT-1-inv-m} in the $n$ direction, then the resulting expressions for $q_{11}$ coincide provided that $\tilde{p}$ and $\tilde{q}$ satisfy system \eqref{Adler-Yamilov}. Moreover, subtracting the two expressions for $p_{00}$ in \eqref{BT-1-inv} we end up with 
$$ \tilde{p}_{10} - \tilde{p}_{01} - \tilde{p}_{00} \left( \tilde{p}_{00} (\tilde{q}_{10} - \tilde{q}_{01}) + \al - \be \right)=0,$$
which holds on solutions of \eqref{Adler-Yamilov}.

Finally, according to Figure \ref{fig:Bianchicomd} and relation \eqref{BT-Bianchi}, the superposition principle for the auto-B\"acklund transformation \eqref{BT-1} follows from 
$$ {\rm{B}}(\tilde{q}_{00},\hat{\tilde{p}}_{00};\varepsilon_2){\rm{B}}(q_{00},\tilde{p}_{00};\varepsilon_1) = {\rm{B}}(\hat{q}_{00},\tilde{\hat{p}}_{00};\varepsilon_1){\rm{B}}(q_{00},\hat{p}_{00};\varepsilon_2),$$
and can be written as
\begin{equation}\label{BT-1-Bianchi}
\hat{\tilde{p}}_{00} = -\,\frac{\tilde{p}_{00} - \hat{p}_{00}}{\varepsilon_1-\varepsilon_2 + (\tilde{p}_{00} - \hat{p}_{00}) q_{00}}, \quad \tilde{q}_{00} - \hat{q}_{00} = \left(\varepsilon_1-\varepsilon_2 + (\tilde{p}_{00} - \hat{p}_{00}) q_{00} \right) q_{00}. 
\end{equation}

\subsubsection*{Second case: $c_1\ne 0$ and $c_2 = 0$} \label{sec:case2}

If $c_2=0$ and $c_1 \ne 0$, we choose $c_1=1$, relations \eqref{eq:det-B-3} imply that $g^{(4)}_{00}= 1$, and the determinant of $\rm{B}$ implies that $g^{(1)}_{00}= p_{00} \tilde{q}_{00}+ \varepsilon$. In view of these choices, we arrive at a Darboux matrix which can be written in terms of matrix \eqref{Darboux-matrix-AY-1} as
\begin{equation}\label{Darboux-matrix-AY-2}
{\rm{C}}(p_{00},\tilde{q}_{00};\varepsilon) = (\lambda+ \varepsilon) \left( {\rm{B}}(\tilde{q}_{00},p_{00};\varepsilon)\right)^{-1},
\end{equation}
whereas relations \eqref{eq:det-B-3a} yield actually system \eqref{BT-1-inv} with the roles of new and old fields interchanged, i.e. system \eqref{BT-1-inv} accompanied by the interchange $(p,q) \leftrightarrow (\tilde{p},\tilde{q})$. Thus in this case we end up with the inverse of the Darboux and B\"acklund transformations we derived previously.

\subsection{Derivation of soliton solutions} \label{sec:solitons}

We employ the auto-B\"acklund transformation \eqref{BT-1} and its superposition principle \eqref{BT-1-Bianchi} in the derivation of soliton solutions of the Adler--Yamilov system \eqref{Adler-Yamilov}. 

\begin{figure}[ht]
	\includegraphics[scale=.8]{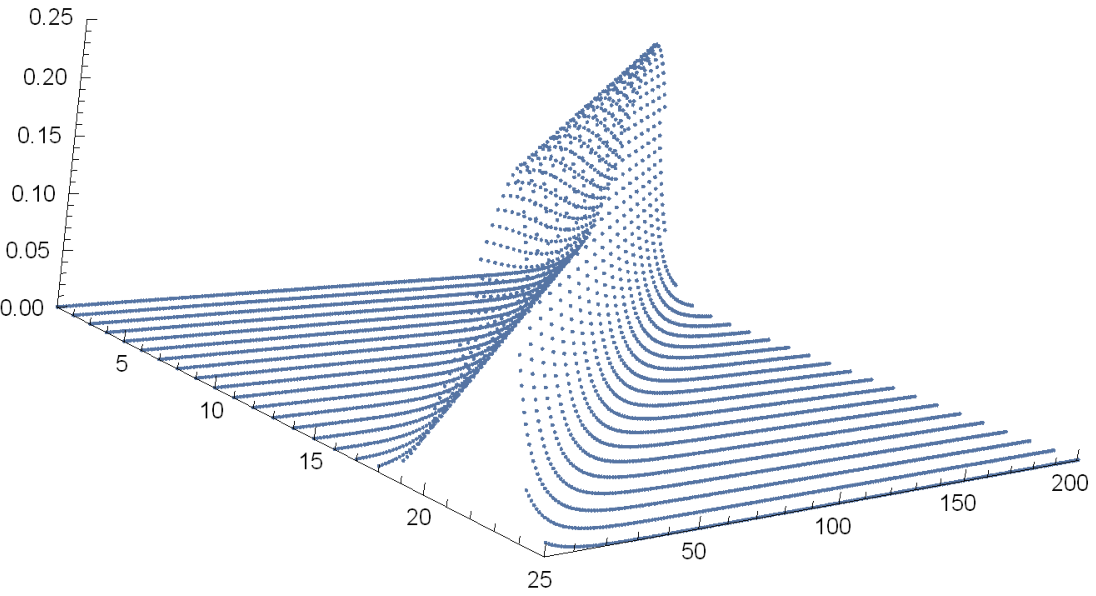} 	\includegraphics[scale=.8]{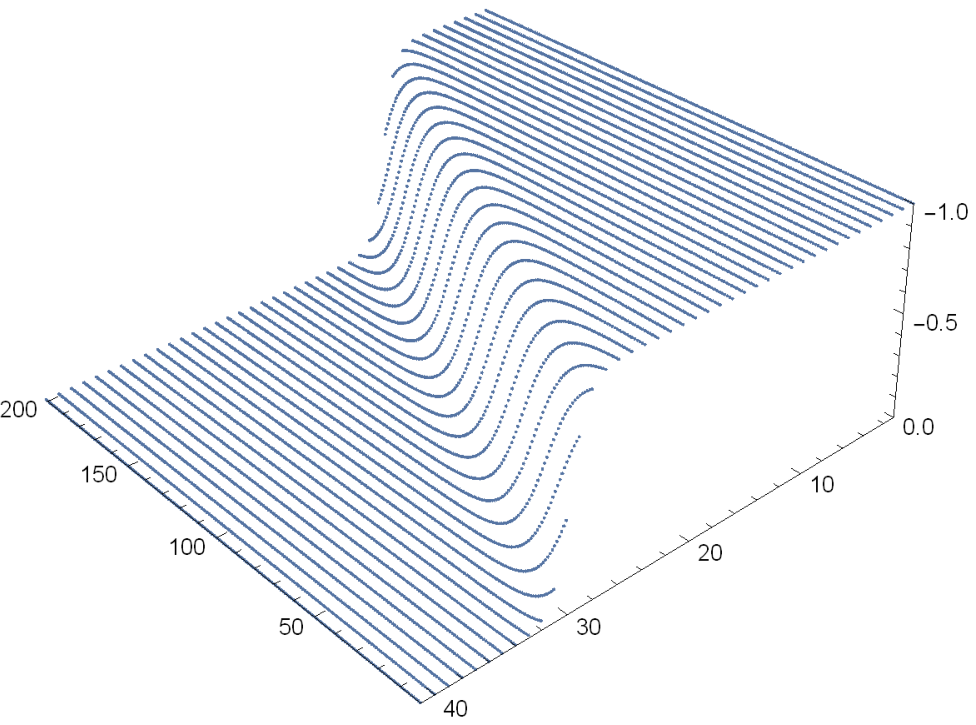}
	\caption{The one soliton solution of the Adler--Yamilov system and the potential $1/g_{00}$. In both cases $\al=8$, $\be=4$, $\varepsilon=1$ and $c= -2$. } \label{fig:kink}
\end{figure}

More precisely, we start with the solution\footnote{This solution can be constructed starting with the zero solution $p_{00}= q_{00}=0$ and using the second transformation we discussed in subsection \ref{sec:DL-AY} with $\varepsilon=0$.} 
\begin{equation} \label{eq:sol-0}
p_{00} = 0, \quad q_{00} = \al^{-n} \be^{-m}.
\end{equation}
With this seed solution, the first equation in \eqref{BT-1-n} and the first one in \eqref{BT-1-m} become
\begin{equation} \label{eq:sol-1}
\tilde{p}_{10}= \frac{\al-\varepsilon}{1+\tilde{p}_{00} \al^{-n-1} \be^{-m}}\tilde{p}_{00}, \quad
 \tilde{p}_{01}= \frac{\be-\varepsilon}{1+\tilde{p}_{00} \al^{-n} \be^{-m-1}}\tilde{p}_{00}.
\end{equation}
We can linearise these Ricatti equations by setting $\tilde{p}_{00} = \al^{n}\be^m/g_{00}$,
$$ (\al - \varepsilon) g_{10} = \al g_{00} + 1, \quad (\be -\varepsilon) g_{01} = \be g_{00} + 1.$$
The general solution of this linear system is
\begin{equation}\label{eq:sol-2}
g_{00} = \frac{\al^n}{(\al - \varepsilon)^n} \frac{\be^m}{(\be - \varepsilon)^m} c - \frac{1}{\varepsilon}, 
\end{equation}
where $c \in {\mathbb{R}}$ is the arbitrary constant of integration, and thus
\begin{subequations}\label{eq:sol-3}
\begin{equation}\label{eq:sol-3a}
\tilde{p}_{00} = \frac{\varepsilon\, \al^n \be^m (\al - \varepsilon)^n (\be-\varepsilon)^m}{c \,\varepsilon\, \al^n \be^m-(\al-\varepsilon)^n (\be-\varepsilon)^m}.
\end{equation}
Using the seed solution \eqref{eq:sol-0} and the updated potential \eqref{eq:sol-3a}, we can use either the equation for $\tilde{q}_{10}$ in \eqref{BT-1-n} or the equation for $\tilde{q}_{01}$ in \eqref{BT-1-m}  to determine $\tilde{q}_{00}$. Both ways lead to
\begin{equation}\label{eq:sol-3b}
\tilde{q}_{00} = \frac{c\, \varepsilon^2}{c \,\varepsilon\, \al^n \be^m-(\al-\varepsilon)^n (\be-\varepsilon)^m}.
\end{equation}
\end{subequations}
This two-parameter family of solutions yields the one-soliton solution of \eqref{Adler-Yamilov}: even though both functions in \eqref{eq:sol-3} diverge, their product represents a soliton. This interpretation is motivated by the relation of the Adler--Yamilov system to the nonlinear Schr\"odinger equation \cite{SPS}, and it is evident from the plot of the product $|\tilde{p}_{00} \tilde{q}_{00}|$. We also plot the potential $1/g_{00}$ which is a kink. See Figure \ref{fig:kink}.

\begin{figure}[ht]
	\includegraphics[scale=.9]{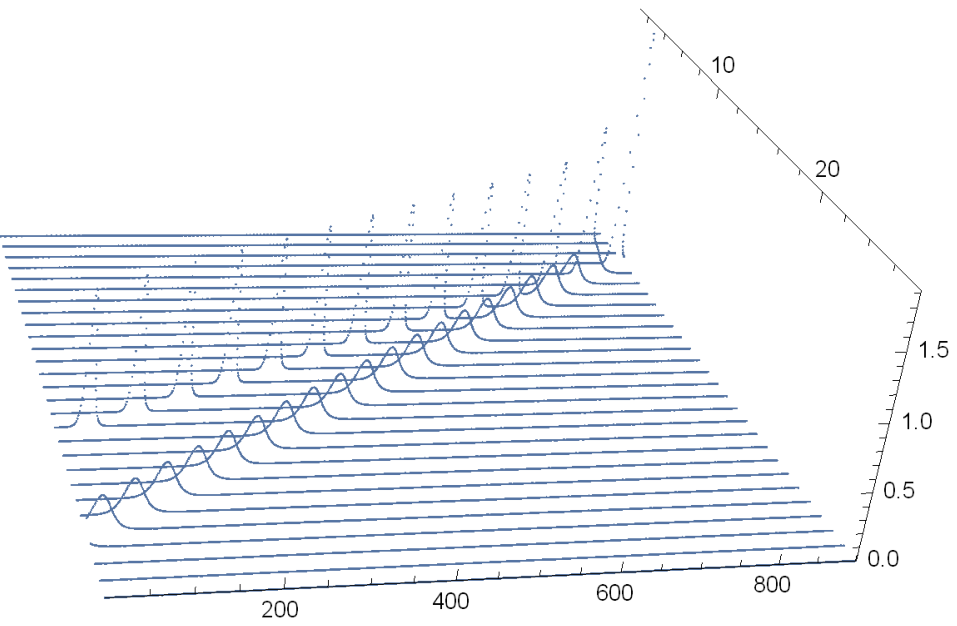} 	\includegraphics[scale=.8]{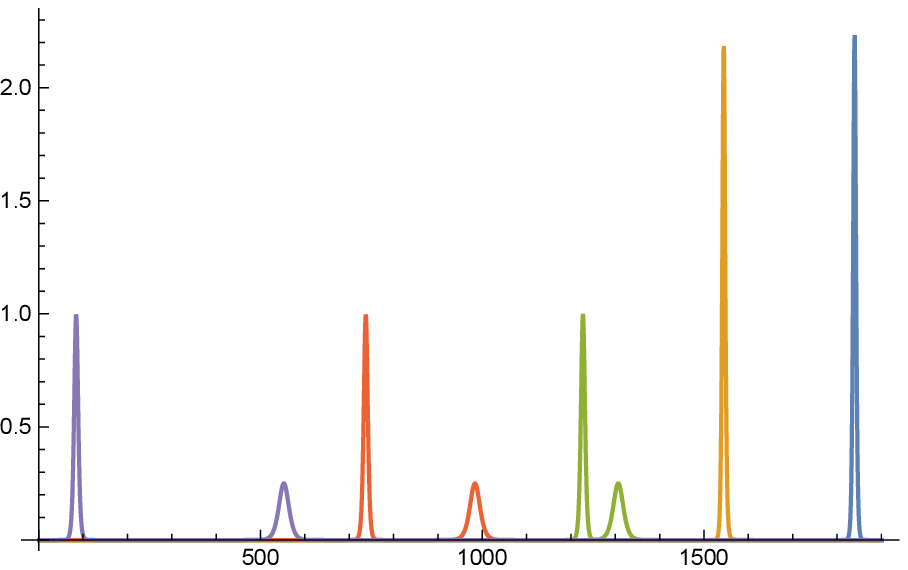}
	\caption{The two-soliton solution of the Adler--Yamilov system: solution \eqref{eq:sol-6} with $\al=8$, $\be=4$, $\varepsilon_1=1$, $c_1= -2$, $\varepsilon_2 = 3$ and $c_2=8$. It should be noted that this solution requires the combination of a non-singular solution, i.e. the pair $(\tilde{p},\tilde{q})$ which corresponds to the one-soliton solution  \eqref{eq:sol-3}, and a singular one which is  $(\hat{p},\hat{q})$. } \label{fig:2sol}
\end{figure}

Having constructed the two-parameter family of solutions \eqref{eq:sol-3}, we may use it along with the superposition principle \eqref{BT-1-Bianchi} to determine a third solution and in particular the two-soliton solution of system \eqref{Adler-Yamilov}. Starting with the same seed solution, the two solutions $(\tilde{p}_{00},\tilde{q}_{00})$ and $(\hat{p}_{00},\hat{q}_{00})$ involved in \eqref{BT-1-Bianchi} follow from \eqref{eq:sol-3} by replacing parameters $(\varepsilon,c)$ with $(\varepsilon_1,c_1)$ and $(\varepsilon_2,c_2)$, respectively. In order to make the presentation more comprehensible, let us introduce the  shorthand notation
$$\delta_0 := \al^n \be^m, \quad \delta_1:= (\al-\varepsilon_1)^n (\be-\varepsilon_1)^m, \quad \delta_2:=  (\al-\varepsilon_2)^n (\be-\varepsilon_2)^m.$$
In terms of this notation, the seed solution is $p_{00}=0$, $q_{00}= 1/\delta_0$, and the two solutions we described can be written as
\begin{equation}\label{eq:sol-5}
\tilde{p}_{00} = \frac{\varepsilon_1 \, \delta_0 \delta_1 }{c_1 \,\varepsilon_1\, \delta_0-\delta_1}, ~~ \tilde{q}_{00} = \frac{c_1\, \varepsilon_1^2}{c_1 \,\varepsilon_1\, \delta_0-\delta_1}, \quad \mbox{and} \quad \hat{p}_{00} = \frac{\varepsilon_2 \, \delta_0 \delta_2}{c_2 \,\varepsilon_2\, \delta_0-\delta_2}, ~~ \hat{q}_{00} = \frac{c_2\, \varepsilon_2^2}{c_2 \,\varepsilon_2\, \delta_0-\delta_2},
\end{equation}
respectively. With these formulae at our disposal, the first relation in \eqref{BT-1-Bianchi} yields 
\begin{subequations}\label{eq:sol-6}
\begin{equation}\label{eq:sol-6a}
\hat{\tilde{p}}_{00} = \frac{\varepsilon_1 \varepsilon_2 \delta_0 \left(c_1 \delta_2 - c_2 \delta_1 \right) + (\varepsilon_1-\varepsilon_2) \delta_1 \delta_2}{c_1c_2 \varepsilon_1\varepsilon_2 (\varepsilon_1-\varepsilon_2) \delta_0 - c_1 \varepsilon_1^2 \delta_2 + c_2 \varepsilon_2^2 \delta_1}.
\end{equation}
To find $\hat{\tilde{q}}_{00}$ we work in the same way we derived \eqref{eq:sol-3b} and according to the Bianchi diagram. More precisely, we can use the second equation either in \eqref{BT-1-n} or in \eqref{BT-1-m} with $\tilde{q}$ replaced by $\hat{\tilde{q}}$, $(p,q)$ replaced by $(\tilde{p},\tilde{q})$ given in \eqref{eq:sol-5}, and parameter $\varepsilon$ replaced by $\varepsilon_2$ (alternatively, we can replace $(p,q)$ with $(\hat{p},\hat{q})$ given in \eqref{eq:sol-5}, and parameter $\varepsilon$ with $\varepsilon_1$), to find
\begin{equation}\label{eq:sol-6b}
\hat{\tilde{q}}_{00} = \frac{c_1 c_2 \varepsilon_1^2 \varepsilon_2^2 (\varepsilon_1-\varepsilon_2)}{c_1c_2 \varepsilon_1\varepsilon_2 (\varepsilon_1-\varepsilon_2) \delta_0 - c_1 \varepsilon_1^2 \delta_2 + c_2 \varepsilon_2^2 \delta_1}.
\end{equation}
\end{subequations}
This four-parameter family of solutions yields the two-soliton solution of \eqref{Adler-Yamilov} in the same way we interpreted \eqref{eq:sol-3} as the one-soliton solution of the Adler--Yamilov system. See the plots of the product $|\hat{\tilde{p}}_{00} \hat{\tilde{q}}_{00}|$ in Figure \ref{fig:2sol}.

\section{Conclusions} \label{sec:con}
In this paper we proposed a new method for deriving B\"acklund transformations and constructing solutions for nonlinear integrable P$\Delta$Es which admit Lax representation but do not necessarily possess the 3D consistency property. Specifically, in our approach we consider Darboux transformations which leave the given Lax pair covariant, and by construction lead to B\"acklund transformations for the corresponding discrete system. The permutability of four Darboux matrices according to the Bianchi diagram in Figure \ref{fig:Bianchicomd} leads to the nonlinear superposition principle of the related B\"acklund transformation. Moreover, the latter transformation and its superposition principle can be used in the construction of interesting solutions to P$\Delta$Es starting from some simple ones. As an illustrative example we used the Adler--Yamilov system \eqref{Adler-Yamilov}. For this system we constructed Darboux and corresponding B\"acklund transformations. With the use of transformation \eqref{BT-1} and its superposition principle \eqref{BT-1-Bianchi} we constructed the one- and two-soliton solutions starting with the seed solution $p_{00} =0, q_{00}= \al^{-n} \be^{-m}$.

In the illustrative example we considered in Section \ref{sec:AY}, the Lax representation \eqref{Lax-rep} involves matrices $\rm{L}$ and $\rm{M}$ which have the same form, see \eqref{Lax-pair-AY}. The natural question arises as to whether our method can be employed to the case of integrable P$\Delta$Es with Lax representation \eqref{Lax-rep} where matrices ${\rm{L}}$ and $\rm{M}$ do not have the same form. The answer to this question is positive, the corresponding transformations may involve auxiliary functions (potentials), and this derivation is similar to the generic construction we presented in subsection \ref{sec:DL-AY}, see Darboux matrix \eqref{eq:det-B-5} and relations \eqref{eq:det-B-2b}, \eqref{eq:det-B-3}, \eqref{eq:det-B-3a}and \eqref{eq:det-B-4}. 

Moreover our considerations can be extended to the generalized symmetries of the discrete system. Generalized symmetries are (integrable) differential-difference equations involving shifts in one lattice direction and are compatible with the P$\Delta$E. Their Lax pair is semi-discrete and its discrete part coincides with the one of the two equations of the fully discrete Lax pair \eqref{linear-sys}, see for instance \cite{Fordy-Pavlos-3}. This relation allows us to extend the Darboux and B\"acklund transformations for the P$\Delta$E to corresponding ones for the differential-difference equations and employ the B\"acklund transformation and its superposition principle in the construction of solutions for the symmetries. We will demonstrate this method and its extensions in our future work in details using the Hirota KdV equation as an illustrative example, as well as systems appeared in \cite{BMX,Fordy-Pavlos,Sasha-Pavlos}.

In fact our results can be used and extended in various ways.
\begin{enumerate}
  \item Apply our method to construct solutions to all the NLS type equations derived in \cite{SPS}.
    
    In \cite{SPS} we classified Darboux transformations related to NLS type equations and constructed integrable discretisations of the latter, namely integrable systems of nonlinear P$\Delta$Es. By employing the discrete Darboux--Lax scheme we proposed in Section \ref{sec:discrete-DL}, one could derive B\"acklund transformations and construct soliton solutions to these nonlinear P$\Delta$Es.
    
    \item Study the solutions of the associated PDEs.
       
    The Adler--Yamilov system \eqref{Adler-Yamilov} constitutes a discretisation of the NLS equation via its Darboux transformation \cite{SPS}. In this paper, we constructed soliton solutions to this system, so one could consider the continuum limits of these solutions to construct solutions to the NLS equation. This procedure could be applied to other NLS type equations which appeared in \cite{SPS}.

    \item Study the corresponding Yang--Baxter maps.
    
    In \cite{Sokor-Sasha, Sokor-Pap} matrix refactorisation problems of Darboux matrices for integrable PDEs were considered in order to derive solutions to the Yang--Baxter equation and the entwining Yang--Baxter equation. Since the generator of Yang--Baxter maps is a matrix refactorisation problem \eqref{BT-Bianchi} it makes sense to understand how B\"acklund transformations for P$\Delta$Es are related to Yang--Baxter maps. In our future work, we plan to show that Yang--Baxter and entwining Yang--Baxter maps are superpositions of B\"acklund transformations of P$\Delta$Es.

    \item Extend the results to the case of discrete systems on a 3D lattice.
    
    One can extend the results employed in this paper to the case of 3D lattice integrable systems. It is expected that the superposition of B\"acklund transformations related to these systems are solutions to the tetrahedron equation.
    
    \item Extend the results to the case of Grassmann algebras.
    
    Grassmann extensions of Darboux transformations were employed in the construction of noncommutative versions of discrete systems, see for instance \cite{Georgi-Sasha, Xue-Levi-Liu} . However it was realised in \cite{Sokor-2020} that quad-graph systems may lose their 3D consistency property in the Grassmann extension. The method we presented here could be generalised and employed in the construction of B\"acklund transformations and the derivation of solutions to Grassmann extended quad-graph systems which appeared in the literature.

\end{enumerate}

\section{Acknowledgements}
 Xenia Fisenko's work was supported by the Ministry of Research and Higher Education (Regional Mathematical Centre ``Centre of Integrable Systems,'' Agreement No. 075-02-2021-1397). Sotiris Konstantinou-Rizos's work was funded by the Russian Science Foundation (grant No. 20-71-10110).

\end{document}